\documentclass[sigconf, nonacm]{acmart}

\usepackage{algorithmic}
\usepackage{algorithm}
 \usepackage{csquotes}

\newcommand\vldbdoi{XX.XX/XXX.XX}
\newcommand\vldbpages{XXX-XXX}
\newcommand\vldbvolume{14}
\newcommand\vldbissue{1}
\newcommand\vldbyear{2020}
\newcommand\vldbauthors{\authors}
\newcommand\vldbtitle{\shorttitle} 
\newcommand\vldbavailabilityurl{http://vldb.org/pvldb/format_vol14.html}
\newcommand\vldbpagestyle{plain} 

\begin{document}
\title{On Distributed Algorithms for Minimum Dominating Set problem, from theory to application}

\author{Sharareh Alipour}
\affiliation{%
  \institution{Institute for research in fundamental sciences(IPM)}
  \streetaddress{P.O. Box 1212}
  \city{Tehran}
  \state{Iran}
}
\email{alipour@ipm.ir}

\author{Ehsan Futuhi}
\affiliation{%
  \institution{Amirkabir university}
  \city{Tehran}
  \country{Iran}
}
\email{futuhiehsan@gmail.com}

\author{Shayan Karimi}
\affiliation{%
  \institution{Amirkabir university}
  \city{Tehran}
  \country{Iran}
}
\email{shayankarimi1376@yahoo.com}

\begin{abstract}
In this paper, we propose a distributed algorithm for the minimum dominating set problem. 
For some especial networks, we prove theoretically that the achieved answer by our proposed algorithm is a constant approximation factor of the exact answer. 
This problem arises naturally in social networks, for example in news spreading, avoiding rumor spreading and recommendation spreading.
So we implement our algorithm on massive social networks and compare our results with the state of the art algorithms.
Also, we extend our algorithm to solve the $k$-distance dominating set problem and experimentally study the efficiency of the proposed algorithm.

Our proposed algorithm is fast and easy to implement and can be used in dynamic networks where the edges and vertices are added or deleted constantly. More importantly, based on the experimental results the proposed algorithm has reasonable solutions and running time which enables us to use it in distributed model practically.
\end{abstract}

\maketitle

\pagestyle{\vldbpagestyle}
\begingroup\small\noindent\raggedright\textbf{PVLDB Reference Format:}\\
\vldbauthors. \vldbtitle. PVLDB, \vldbvolume(\vldbissue): \vldbpages, \vldbyear.\\
\href{https://doi.org/\vldbdoi}{doi:\vldbdoi}
\endgroup
\begingroup
\renewcommand\thefootnote{}\footnote{\noindent
This work is licensed under the Creative Commons BY-NC-ND 4.0 International License. Visit \url{https://creativecommons.org/licenses/by-nc-nd/4.0/} to view a copy of this license. For any use beyond those covered by this license, obtain permission by emailing \href{mailto:info@vldb.org}{info@vldb.org}. Copyright is held by the owner/author(s). Publication rights licensed to the VLDB Endowment. \\
\raggedright Proceedings of the VLDB Endowment, Vol. \vldbvolume, No. \vldbissue\ %
ISSN 2150-8097. \\
\href{https://doi.org/\vldbdoi}{doi:\vldbdoi} \\
}\addtocounter{footnote}{-1}\endgroup

\ifdefempty{\vldbavailabilityurl}{}{
\vspace{.3cm}
\begingroup\small\noindent\raggedright\textbf{PVLDB Artifact Availability:}\\
The source code, data, and/or other artifacts have been made available at \url{\vldbavailabilityurl}.
\endgroup
}
\section{Introduction}
Nowadays online social networks are growing exponentially and they have important effect on our daily life. They influence politics and economics. Online shopping, online advertisement and social medias have important role in our life style. Even we communicate more with friends and family on social networks than face to face meetings.

To influence the network participants a key feature in a social network is the ability to communicate quickly within the network. For example, in an emergency situation, we may need to be able to reach to all network nodes, but only a small number of individuals in the network can be contacted directly due to the time or other constraints. However, if all nodes from the network are connected to at least one such individual who can be contacted directly (or is one of those individuals) then the emergency message can be quickly sent to all network participants.
Or suppose that in a social network some nodes should be selected to distribute a news in the network or should be selected to avoid spreading rumors by checking the spreading news in the network.
In these scenarios the goal is to choose the minimum number of such nodes. 
In graph theory this problem is called minimum dominating set problem.

Given a graph $G = (V,E)$  with the vertex set $V$ and the edge set $E$, a subset $S \subseteq V$ is a dominating set for $G$ if each node $v\in V\backslash S$ has a neighbor in $S$.
Also $S\subseteq V$ is a total dominating set of $G$ if each node $v\in V$ has a neighbor in $S$. 
Let $\gamma(G)$ and  $\gamma_t(G)$ be the size of a minimum dominating set (MDS) and  a minimum total dominating set (MTDS) for a graph $G$, respectively.
It is easy to see that if $G$ has no isolated vertices then $\gamma(G)\leq \gamma_t(G)\leq2\gamma(G)$. 
Dominating set is an important concept in graph theory which arises in many areas. When the given network is very large such as social networks, we have limitations in memory and time. Sometimes we need to run the algorithm on distributed model.
In a distributed model the network is abstracted as a simple $n$-node undirected graph $G = (V,E)$. There is one processor on each graph node $v \in V$, with a unique $\Theta(log n)$-bit identifier $ID(v)$, who initially knows only its neighbors in $G$. Communication happens in synchronous rounds. Per round, each node can send one, possibly different, $O(\log n)$-bit message to each of its neighbors. At the end, each node should know its own part of the output. For instance, when computing the dominating set, each node knows whether it is in the dominating set or has a neighbor in the dominating set.  

An extension of minimum dominating set  problem is minimum $k$-distance dominating set problem where the goal is to choose a subset $S\subseteq V$ with minimum cardinality such that for every vertex 
$v\in V\backslash S$, there is a vertex $u\in S$ such that there is a path between them of length at most $k$. The minimum total $k$-distance dominating set is defined similarly. 
In social networks the minimum $k$-distance dominating set can be considered as social recommenders. 
The close nodes influence each other and they have the same preferences in a network.
Suppose that we want to give a recommendation on a special product (e.g. which movie to watch) to each node of the network but we can not reach all of them because of the time constraint and advertising cost. We may choose minimum number of nodes such that they dominate all other nodes within distance $k$ from them. We give a recommendation to each of the selected nodes and then they spread it in the network. This is equal to solving the $k$-distance dominating set problem. For more on social recommendation see \cite{avn}. 

\subsection{Previous works}

Finding a minimum dominating set is NP-complete \cite{karp}, even for planar graphs of maximum degree $3$ \cite{gar},
and cannot be approximated for general graphs with a constant ratio under the assumption $P\neq NP$ \cite{raz}.
An $O(log n)$-approximation factor can be found by using a simple greedy algorithm. Moreover, negative results have been proved for the approximation of MDS even when limited to the power law graphs \cite{gas}.

A number of works have been done on exact algorithms for MDS, which mainly focus on improving the upper bound of running time. State of the art exact algorithms for MDS are based on the branch and reduce paradigm and can achieve a run time of $O({1.4969}^n)$ \cite{van}. Fixed parameterized algorithms have allowed to obtain better complexity results \cite{kar}. The main focus of such algorithms is on theoretical aspects.
In the distributed model it is known that finding small dominating sets in local model is hard. Kuhn and Wattenhofer \cite{kuhn} proposed an algorithm with an approximation factor that is  a function of the number of communication rounds. And the number of communication rounds also depends on $\Delta$. 
For more theoretical results on distributed algorithms for MDS problem see \cite{akh}.
 In \cite{hike} it has been shown that for any $\epsilon>0$ there is no deterministic local algorithm that finds a $(7-\epsilon)$-
approximation of a minimum dominating set for planar graphs. However, there exist an algorithm with approximation factor of $52$ for computing a MDS in planar graphs \cite{and,chris} in local model and an algorithm with approximation factor of $636$ for anonymous networks \cite{and,waw}. In \cite{ali}, they improved the approximation factor in anonymous networks to $18$ in planar graphs without $4$-cycles. For more information on local algorithms see \cite{suo}.

In practice, these theoretical algorithms are not applicable specially in large social networks because of time and space constraints. So we need to use heuristic algorithms  to obtain solutions. See \cite{san} for a  comparison among several greedy heuristics for MDS. 
Heuristic search methods such as genetic algorithm \cite{her} and ant colony optimization \cite{pot11,pot13} have been developed to solve MDS. Also Hyper metaheuristic algorithms combine different heuristic search algorithms and preprocessing techniques to obtain better performance \cite{pot13,sach,blum,gen,abu}. 
These algorithms were tested on standard benchmarks with up to thousand vertices. The configuration checking (CC) strategy \cite{cai11} has been applied to MDS and led to two local search algorithms. Wang et al. proposed the CC2FS algorithm for both unweighted and weighted MDS \cite{w}, and obtained better solutions than ACO-PP-LS  \cite{pot13} on standard benchmarks. Afterwards, another CC-based local search named FastMWDS was proposed, which significantly improved CC2FS on weighted massive graphs \cite{wang18}. Chalupa proposed an order-based randomized local search named RLSo \cite{chal}, and achieved better results than ACO-LS and ACO-PP-LS \cite{pot11,pot13} on standard benchmarks of unit disk graphs as well as some massive graphs. Fan et. al. designed a local search algorithm named ScBppw \cite{fan}, based on two ideas including score checking and probabilistic random walk. Recently an efficient local search algorithm for MDS is proposed in \cite{cai20}. The algorithm named FastDS is evaluated on some standard benchmarks. FastDS obtains the best performance for almost all benchmarks, and obtains better solutions than previous algorithms on massive graphs in their experiments.

A recent study for the $k$-distance dominating set problem can be found in \cite{min}. They proposed a heuristic algorithm that can handle real-world instances with up to $17$ million vertices and $33$ million edges. They stated that this is the first time such large graphs are solved for the minimum $k$-distance dominating set problem. They compared their proposed algorithm with the other best know algorithms for this problem.

\subsection{Our results}
 
In this paper we propose a local approximation algorithm which is an extension of \cite{ali}. 
We prove that the approximation factor of this algorithm in planar triangle free graphs is $16$ and $32$ for MTDS problem and MDS problem, respectively. 

We implement the centralized version of our algorithm and run it on real massive networks.
The proposed algorithm is fast and easy to implement both in distributed model and centralized model and the achieved results are satisfactory. Also the proposed algorithm can be used in dynamic networks, where the network is changing constantly i.e. during the time some nodes and edges are added or deleted. For example in a social network some users may be offline during a time period and so they are disconnected from the network and this affects the network.
We compare the results of centralized version of our algorithm with the results of \cite{cai20} to show the efficiency of the achieved results.
Also since the algorithm is fast, we compute dominating set for the other massive graphs.
For further experiments we modify the algorithm to solve the $k$-distance dominating set problem and compare the results with \cite{min}.

Note that our main goal is to show the efficiency of proposed algorithm in distributed model and the experimental results show that the algorithm can be used in practice for large networks.

\section{LOCAL algorithm for dominating set in graphs}
In this section we present our local algorithm for computing a (total) dominating set in graphs. We also compute the approximation factor of the algorithm, Algorithm \ref{alg1}, for triangle free planar graphs.

\subsection{Algorithm}

Let $G$ be a graph with vertex set $V=\{v_1,\dots, v_n\}$ and let $d_i$ be the degree of $v_i$. A local distributed algorithm that computes a total dominating set for $G$ is presented in Algorithm \ref{alg1}. We assume that $G$ does not have isolated vertices (vertices with degree $0$) because the concept of total dominating set can not be defined if the graph has a vertex with degree $0$.

\begin{algorithm}
\begin{algorithmic}[1]
\label{alg1}
\caption{Distributed Algorithm for computing a total dominating set in a graph with given integer $m\ge 0$.}
\label{alg1}
\STATE In the first round, each node $v_i$ chooses a random number $0<r_i<1$ and computes its weight $w_i=d_i+r_i$ and sends $w_i$ to its adjacent neighbors.
\STATE In the second round, each node $v$ marks a neighbor vertex $v_i$, whose weight $w_i$  is maximum among all the other neighbors of $v$.
\FOR {$m$ rounds} 
  \STATE Let $x_i$ be the number of times that a vertex is marked by its neighbor vertices, let $w_i=x_i+r_i$
  \STATE Unmark the marked vertices.
  \STATE Each vertex marks the vertex with maximum $w_i$ among its neighbor vertices.
    \ENDFOR
\STATE The marked vertices are considered as the vertices in our total dominating set for $G$.

\end{algorithmic}
\end{algorithm}

Now we show the correctness of Algorithm \ref{alg1}. In each step each vertex marks one of its adjacent vertices, this means the marked vertices form a total dominating set. 
Note that a total dominating set is also a dominating set so, this algorithm also serves as an algorithm for MDS problem.
 
\subsection{Approximation factor of Algorithm \ref{alg1} for $m=0$ in planar graphs without triangles}

Now we compute the approximation factor of Algorithm \ref{alg1} for planar triangle free graphs. 
\begin{theorem}
\label{main}
The approximation factor of Algorithm \ref{alg1} with $m=0$ for triangle free planar graphs for MDS problem and MTDS problem are respectively $32$ and $16$.
 \end{theorem}
\begin{proof}
Let $V'=\{v'_1, v'_2, \dots, v'_k\}$ be the set of vertices that are marked in the algorithm. Let $V_{opt}=\{v^*_1,\dots, v^*_{opt}\}$ be an optimal solution for MTDS. We construct a planar multi-graph $G'$ (i.e. might have multiple edges between two of its vertices) with the vertex set of $V_{opt}$ and $k$ edges that are traced either along an edge of $G$ or along a path of length $3$ of $G$ (called a broken edge). We show that each edge of this graph is counted at most twice, so we have a planar graph with at least $k/2 $ edges.  Since this graph is planar, we can divide the edges between two vertices (if there are many) into two groups, the top and bottom edges (called outer edges), and the middle edges (called inner edges). We can at most have $6opt$ outer edges (planar simple graph with $m$ vertices has at most $3m$ edges. Our graph might have double edges, hence the factor of $6$). To any inner edge we associate a vertex $z^*$ in $V_{opt}$ and show that each element of $V_{opt}$ can be associated to at most $2$ inner edges and hence the number of inner edges is at most $2opt$. So $k/2$ is at most $6opt+2opt=8opt$ and $k< 16opt$. Now we explain the details.

To any vertex $v$ associate in an arbitrary way a vertex $v^* \in V_{opt}$ that dominates it. For each $v'\in V'$, choose a vertex $u$ that marks it. Now let $v^*$ and $u^*$ be the dominating vertices in $V_{opt}$ for $v'$ and $u$. Since $G$ is triangle free, so $v^*\neq u^*$ . If $u^*$ and $v^*$ are different from $v'$ and $u$
then we draw a broken edge traced on the path of length three $u^*-u-v'-v^*$ otherwise they are connected by an edge $e$ in $G$ that we trace $e$ (tracing means we draw an edge almost identical, so we avoid unnecessary collision of edges).
For each element of $V'$ we have an edge between two vertices of $V_{opt}$ so we have $k$ edges. The broken edge $u^*-u-v'-v^*$ will be counted twice if $u$ is also in $V'$ and $v'$ is the vertex that marks it. There is no crossing, since a crossing can occur either at the vertex $u$ or $v'$, when $u$ is in $V'$ or $v'$ marks a vertex $v''$ in $V'$. In both cases since we have chosen
a fixed dominating vertex in $V_{opt}$, we can draw the broken edge in a non-crossing manner. If $v^*-x-y-w^*$ is an inner edge and $x$ has marked $y$, since 
$d_{v^*}$ is at least $3$ (there are two outer edges in this case) so $d_y$ is at least $3$ as well, and hence there is a vertex $z$ that is connected to $y$ and hence there is a vertex $z^*$ in $V_{opt}$ that dominates it.
   We associate $z^*$ to this inner edge. Since $z^*$ is bounded between two (broken) edges from $v^*$ to $w^*$ it can be associated to at most one more inner edge, therefore the number of inner edges is at most $2opt$.

 \end{proof}
 
 Note that $r_i$'s are not used in this proof, however in experiments they give a total order on the vertices of the graph (according to $d_{v_i}+r_i$) and hence the number of marked vertices will be reduced in principle. For example 
consider a complete graph (although non-planar) without $r_i$'s then it is possible that each vertex chooses a different neighbor and instead of just one vertex, all the vertices get marked. But with $r_i$'s all the vertices will choose one vertex.
 
Now we show how to extend this algorithm to solve the $k$-distance dominating set problem. In solving the $k$-distance dominating set problem, for each node $v$, its neighbors are the vertices that their distance from $v$ is at most $k$ and the algorithm is run as before.

\section{Experimental results}
In this section we present our experimental results. 
In theory we have seen that the algorithm has good approximation factor in certain networks, for example planar graphs without 4-cycles or planar graphs without triangles. 
These approximation factors are achieved in the worst case, however in practice this is not the case that happens most of the time. 
So based on this fact we can expect good results for real data too. 

We choose $5$ benchmarks which were used in \cite{cai20} to compare the performance of our algorithms with their algorithm which is denoted as FastDS. In \cite{cai20} they compared  FastDS with four heuristic algorithms, including CC2FS \cite{w}, FastMWDS \cite{wang18}, RLS \cite{chal} and ScBppw \cite{fan}. It is stated that CC2FS is good at solving standard benchmarks, while FastMWDS and ScBppw are designed to solve massive graphs, and RLSo is a recent algorithm that outperforms previous ant optimization and hyper meta-heuristic algorithms.

In \cite{cai20} for each instance, all algorithms were executed $10$ times with random seeds $1,2,3\dots10$. The time limit of each run was $1000$ seconds. For each instance, they reported the best size (Dmin) and the average size (Davg) of the solutions found over the $10$ runs. 
Compared to the other algorithms the FastDS algorithm outperformed in the quality of solutions in most cases \cite{cai20}.

In the following we present a brief description of the benchmarks from \cite{cai20}.
 
T1\footnote{http://mail.ipb.ac.rs/~rakaj/home/BenchmarkMWDSP.htm}: This data set consists of $520$ instances where each instance has two different weight functions. As in \cite{cai20} we select these original graphs where the weight of each vertex is set to $1$. There are $52$ families, each of which contains $10$ instances with the same size.

BHOSLIB\footnote{http://networkrepository.com/bhoslib.php}: This benchmark are generated based on the RB model near the phase transition. It is known as a popular benchmark for graph theoretic problems.

SNAP\footnote{http://snap.stanford.edu/data}: This benchmark is from Stanford Large Network Dataset Collection. It is a collection of real world graphs from $10^4$ vertices to $10^7$ vertices.

DIMACS10\footnote{http://networkrepository.com/dimacs10.php}: This benchmark is from the 10th DIMACS implementation challenge, which aims to provide large challenging instances for graph theoretic problems.

Network Repository\footnote{http://networkrepository.com/}: The Network Data Repository includes massive graphs from various areas. Many of the graphs have $100$ thousands or millions of vertices. This benchmark has been widely used for graph theoretic problems including vertex cover, clique, coloring, and dominating set problems.

As in \cite{cai20} for SNAP benchmarks we consider the graphs with at list $30000$ vertices and for Repository benchmark we choose the graphs with at least $10^5$ vertices.

In \cite{cai20} the algorithms were implemented in C++ and complied by g++. All experiments were run on a server with Intel Xeon E5-2640 v4 2.40GHz with 128GB RAM under CentOS 7.5. In our experiment the algorithm is implemented in Java and is run on a server with Intel Xeon E5-2650 v3s 2.29GHz with 70GB RAM under Ubuntu 18.04.4.

In Table \ref{time} the average running time of Algorithm \ref{alg1},  FastDS, CCFS and FastMWDS for BHOSLIB and T1 benchmarks are stated. Note that we can not compare the running times because the configuration of our system and the programming language in our experiments are not the same as \cite{cai20} 
but we remark that their's is more efficient than ours.

 As it can be seen for these benchmarks our algorithm is really fast. In all our experiments we mention the running time for each instance in our sequential implementations.

\begin{table}[h]
\begin{footnotesize}
\caption{Average running time of algorithms.}
\begin{center}
\begin{tabular}{| lllll |}
 \hline
 
\textbf{Benchmark}& \textbf{CCFS}&\textbf{FastMWDS}& \textbf{  FastDs} &  \textbf{  Algorithm \ref{alg1} $m=0,2,5$}\\ \hline
T1  &   4.88s& 8.35s&             11.23s &      0.0021, 0.0048, 0.0074\\ \hline
BHOSLIB  &96.68s&  101.44s&    95.62s   &    0.028, 0.076, 0.111\\ \hline

\end{tabular}
\end{center}
\label{time}
\end{footnotesize}
\end{table}

In the following, we presented our experimental results. Note that our algorithm computes a total dominating set and since a total dominating set is also a dominating set, we report the marked vertices as a dominating set. In theory, the MTDS is at most $2$ times of MDS, so this explains why our answers are greater than the answers of FastDS algorithm. 
Also, note that we can have a trade-off between the running time and the quality of solution by changing the value of $m$. However our experiments show that the quality of solution does not have significant improvement for large values of $m$. By experiment we have chosen the value of $m$ to be $0$, $2$ and $5$.

Table \ref{t1} contains the results for T1 benchmark. In Table \ref{bho} we present the results for BHOSLIB benchmark. In this benchmark the achieved results by Algorithm \ref{alg1} are surprisingly better than FastDs. This happens because of the nature of our algorithm. In dense graphs or the graphs with large maximum degree, close to $n$ (number of vertices), the adjacent vertices to the maximum degree choose it, so the number of marked vertices is small and close to the exact solution.
For example in frb100-40 instance of BHSLIB benchmark, the number of vertices is about 4000, the maximum degree is 3864, the minimum degree is 3553. Our algorithm choose $3$ vertices, v1096, v1555 and v1088 with degrees $3864$, $3861$ and $3861$ respectively.

Table \ref{snap} contains the result for snap and DIMACS10 benchmarks. In Table \ref{nr} we present the result for Network repository benchmark. In this benchmark we also run the algorithm on more instances related to web and social networks other than \cite{cai20}. The empty cells are the instances that were not computed in \cite{cai20}.

\begin{table}[h]
\begin{footnotesize}
\caption{Experimental results for T1 benchmark for $m=0,2,5$.}
\begin{center}
\begin{tabular}{| llll |}
 \hline
     
\textbf{instance}&                  \textbf{time(s) $m=0,2,5$ }&              \textbf{Sol Alg \ref{alg1}} &                                      \textbf{ FastDS \cite{cai20}}\\ \hline
V100E100 &                  0.001, 0.002, 0.004&            66, 65, 63&                 33.6\\ \hline

V100E1000  &              0.000, 0.000, 0.001&            14, 12, 12  &              7.5\\ \hline

V100E2000&                 0.000, 0.000, 0.001&            6, 6, 6  &                  4.1\\ \hline

V100E250 &               0.000, 0.001, 0.003  &          39, 33, 31&                 19.9\\ \hline

V100E500 &                0.001, 0.002, 0.005&            22, 19, 20&                 12.2\\ \hline

V100E750 &                0.001, 0.003, 0.005&            14, 14, 12 &                9\\ \hline

V150E1000 &                0.001, 0.001, 0.001&            25, 22, 22 &                15\\ \hline

V150E150&                   0.001, 0.003, 0.003&            100, 97, 96&                50\\ \hline

V150E2000  &                0.000, 0.001, 0.004&            16, 15, 14 &                9\\ \hline

V150E250&                   0.001, 0.002, 0.012&            66, 63, 61 &                39.1\\ \hline

V150E3000 &                 0.001, 0.001, 0.005&            14, 13, 11&                 6.9\\ \hline

V150E500&                  0.000, 0.005, 0.006&            39, 37, 36  &               24.6\\ \hline

V150E750 &                  0.000, 0.001, 0.002&            34, 27, 26&                 18.3\\ \hline

V200E1000&                  0.000, 0.001, 0.003&            44, 36, 35&                 24.4\\ \hline

V200E2000 &                 0.000, 0.001, 0.002&            26, 22, 22 &                15\\ \hline

V200E250 &                  0.001, 0.002, 0.003 &           102, 99, 97&                61.1\\ \hline

V200E3000 &                0.001, 0.001, 0.002&            21,16,15 &                11\\ \hline

V200E500 &                  0.000, 0.001, 0.005&            65,61,60 &                36.6\\ \hline

V200E750 &                  0.001, 0.001, 0.002&            54,47,45 &                30\\ \hline

V250E1000&                 0.001, 0.001, 0.003&            61,54,52 &                36\\ \hline

V250E2000&                 0.000, 0.001, 0.002 &           38, 31, 31 &                21.6\\ \hline

V250E250  &                 0.002, 0.002, 0.004 &           164, 160, 160 &             83.3\\ \hline

V250E3000&                  0.001, 0.003, 0.005&            33, 28, 27  &               16\\ \hline

V250E500  &                 0.000, 0.001, 0.002&            96, 93, 91&                 57.8\\ \hline

V250E5000 &                 0.002, 0.003, 0.008 &           20, 19, 18&                 11\\ \hline

V250E750&                   0.001, 0.002, 0.003 &           79, 68, 67&                 44\\ \hline

V300E1000&                  0.000, 0.002, 0.003&            89,  79, 77 &                48.6\\ \hline

V300E2000 &                 0.001, 0.001, 0.004&            59, 54, 48 &                29.4\\ \hline

V300E300 &                  0.001, 0.004, 0.005 &           193, 192, 192&              100\\ \hline

V300E3000 &                 0.001, 0.003, 0.007&            42, 35, 33&                 22\\ \hline

V300E500&                   0.001, 0.002, 0.003 &           138, 125, 122 &             77.7\\ \hline

V300E5000  &               0.001, 0.005, 0.010 &           31, 27, 23 &                15.1\\ \hline

V300E750 &                  0.001, 0.004, 0.005&            114, 96, 93&                59.6\\ \hline

V500E1000&                  0.002 ,0.004, 0.006&            199, 187, 181&              114.7\\ \hline

V500E10000 &                0.005, 0.008, 0.012&            41,  40, 37 &                22.2\\ \hline

V500E2000  &                0.002, 0.008, 0.012  &          134, 116, 114 &             71.2\\ \hline

V500E500 &                  0.003, 0.006, 0.013 &           329, 326, 326&              167\\ \hline

V500E5000 &                0.004, 0.011, 0.021&            67, 62, 60  &               36.9\\ \hline

V800E100 &                 0.005, 0.014, 0.015 &           537, 528, 523&              267\\ \hline

V800E1000 &               0.004, 0.011, 0.013 &           434, 421, 414  &            242.5\\ \hline

V800E10000&               0.003, 0.009, 0.013&            91, 83, 78&                 50.2\\ \hline

V800E2000  &               0.014, 0.015, 0.013 &           282, 266, 252&              158.3\\ \hline

V800E5000 &               0.004, 0.012, 0.015 &           145, 138, 130 &             82.6\\ \hline

V1000E1000 &             0.010, 0.012, 0.014 &           677, 654, 654 &             333.7\\ \hline

V1000E10000&            0.005, 0.012, 0.022 &           145, 129, 118 &             74\\ \hline

V1000E15000 &           0.006, 0.013, 0.020 &           100, 90, 87 &               55\\ \hline

V1000E20000 &           0.007, 0.018, 0.020 &           86, 75, 65 &                45\\ \hline

V1000E5000 &            0.006, 0.012, 0.014 &           226, 201, 197&              121.1\\ \hline

 \end{tabular}
\end{center}
\label{t1}
\end{footnotesize}
\end{table}

\begin{table*}[h]
\begin{footnotesize}
\caption{Experimental results for BHOSLIB benchmark for $m=0,2,5$.}

\begin{tabular}{|llll | llll     |}
 \hline
 
\textbf{instance}&                  \textbf{time(s) $m=0,2,5$ }&              \textbf{Sol Alg \ref{alg1}} &                                      \textbf{ FastDS \cite{cai20}}& \textbf{instance}&                  \textbf{time(s) $m=0,2,5$ }&              \textbf{Sol Alg \ref{alg1}} &                                      \textbf{ FastDS \cite{cai20}}\\ \hline

frb40-19-1&           0.024, 0.032, 0.115&       3, 3, 3&               14& frb50-23-4 &          0.023,  0.067, 0.089&       7, 4, 4 &              18\\ \hline

frb40-19-2 &          0.023, 0.032, 0.061&       6, 3, 3&               14&frb50-23-5 &          0.022, 0.075, 0.084 &      5, 3, 3 &              18\\ \hline

frb40-19-3&          0.003, 0.034, 0.059&       4, 4, 4 &              14&frb53-24-1&           0.032, 0.077, 0.094 &      3, 3, 3&               19\\ \hline

frb40-19-4&             0.023, 0.033, 0.084&       3, 4, 4 &              14&frb53-24-2 &          0.033, 0.096, 0.115&       4, 3, 3 &              19\\ \hline

frb40-19-5 &            0.019, 0.035, 0.074&       3, 3, 3&               14&frb53-24-3 &          0.032, 0.077, 0.094&       3, 3, 3&               19\\ \hline

frb45-21-1 &        0.022, 0.092, 0.098&       5, 5, 4&               16&frb53-24-4 &          0.026, 0.086, 0.091&       7, 6, 6&               18\\ \hline

frb45-21-2 &          0.024, 0.045, 0.077&       7, 6, 5&               16&frb53-24-5&             0.029, 0.081, 0.098&       5, 3, 3&               19\\ \hline

frb45-21-3 &         0.022, 0.066, 0.078&       3, 3, 3&               16&frb59-26-1&       0.047, 0.109, 0.118&       3, 3, 3 &              20\\ \hline

frb45-21-4 &           0.022, 0.047, 0.076&       4, 4, 4&               16&frb59-26-2&          0.029, 0.087, 0.107&       3, 3, 3&               21\\ \hline

frb45-21-5 &        0.032, 0.042, 0.075&       3, 3, 3 &              16&frb59-26-3 &           0.046, 0.114, 0.119 &      6, 3, 3 &              21\\ \hline

frb50-23-1 &           0.024, 0.064, 0.098&       4, 4, 4 &              18&frb59-26-4 &           0.031, 0.080, 0.155&       3, 3, 3&               21\\ \hline

frb50-23-2 &         0.023, 0.085, 0.087&       4, 4, 4 &              18&frb59-26-5 &           0.025, 0.099, 0.102&       6, 3, 3&               21\\ \hline

frb50-23-3 &           0.007, 0.063, 0.084&       5, 3, 3 &              18&frb100-40 &            0.083,  0.255, 0.550&       3, 3, 3&               36\\ \hline

\end{tabular}

\label{bho}
\end{footnotesize}
\end{table*}

\begin{table*}[h]
\begin{footnotesize}
\caption{Experimental results for Network snap and DIMACS10 benchmark for $m=0,2,5$.}
\begin{center}
\begin{tabular}{| llll |}
 \hline
\textbf{instance}&                  \textbf{time(s) $m=0,2,5$ }&              \textbf{Sol Alg \ref{alg1}} &                                      \textbf{ FastDS \cite{cai20}}\\ \hline

Amazon0302(V262K E1.2M)&     				                4.322, 11.832, 21.681 &        82690, 71023, 68914&               35593\\ \hline
Amazon0312(V400K E3.2M) &          			                7.833,19.076, 32.826&         100960, 81107, 75368&              45490\\ \hline
Amazon0505(V410K E3.3M) &           				         6.962,16.653, 29.229 &        91009, 78355, 72273&               47310\\ \hline
Amazon0601(V403K E3.3M)&            				            6.844,15.708, 26.702 &        85559, 70372, 66267&               42289\\ \hline
email-EuAll(V265K E420K)&      				 	            1.156, 2.208, 3.613&           32007, 31935, 31925&               18181\\ \hline
p2p-Gnutella24(V26K E65K) &       				               0.039, 0.117, 0.208&           6056, 5724, 5569&                  5418\\ \hline
p2p-Gnutella25(V22K E54K)&           			                 0.033, 0.084, 0.150&           5066, 4771, 4637&                  4519\\ \hline
p2p-Gnutella30(V36K E88K)&         				                0.065, 0.201, 0.308&           8067, 7559, 7379&                  7169\\ \hline
p2p-Gnutella31(V62K E147K)&      				               0.174, 0.496, 1.057 &          13917, 13123, 13086&               12582\\ \hline
soc-sign-Slashdot081106(V77K E516K)&			                0.198, 0.487, 0.845&           16197, 15184, 15069&               14312\\ \hline
soc-sign-Slashdot090216(V81K E545K)&			               0.198, 0.487, 0.845&           17291, 16290, 15871&               15305\\ \hline
soc-sign-Slashdot090221(V82K E549K)&			                0.381, 0.817, 1.383&           16924, 16331, 16117&               -\\ \hline
soc-Epinions1(V75K E508K) &          			              2.050, 2.589, 3.298 &          17146, 16646, 16132&               15734\\ \hline
web-BerkStan(V685K E7.6M)&      				               11.522, 28.080, 52.323&        68449, 60265, 55912&               28432\\ \hline
web-Stanford(V281K E2.3M)&        				           1.191, 2.824, 5.330 &          36928, 32956, 30845&               13199\\ \hline
wiki-Talk(V2.3M E5M)  &              				        5.447, 18.860, 35.098&         40303, 39194, 39191&               36960\\ \hline
web-NotreDame(V325K E1.5M)&         			           2.301, 6.293, 11.914&          34788, 31929, 30883&               23735\\ \hline
wiki-Vote(V7K E103K) &               			                  2.155, 2.419, 2.663&           1381, 1189, 1183 &                 1116\\ \hline
cit-HepPh(V34K E421K)&            				              16.367, 18.854, 22.278&        7476, 4866, 4305&                  3078\\ \hline
cit-HepTh(V27K E352K)&            			                       11.972, 16.458, 21.249&        6159, 4446, 4155 &                 2936\\ \hline
rgg-n-2-17-s0 &               1.101, 2.649, 4.817 &	         33789, 24020, 23595 &        43412\\ \hline
rgg-n-2-19-s0 &             14.948, 37.951, 69.568 &         127417, 88424, 86640&       44423\\ \hline
rgg-n-2-20-s0 &               57.436,143.619, 265.936&       248861, 170092, 166453&     84708\\ \hline
rgg-n-2-21-s0	&           222.205, 554.531, 1021.933&     487626, 328695, 320941&     162266\\ \hline

rgg-n-2-22-s0&	           862.451, 2127.245, 3918.068&     952637, 635144, 619729&    312350\\ \hline

rgg-n-2-23-s0&	          3387.642, 8272.262, 15166.054&   1866988, 1230026, 1199105&    605278   \\ \hline

citationCiteseer &        1.818, 5.231, 9.883&            58985, 52797, 52512&       43412\\ \hline

coAuthorsCiteseer &          1.322,  3.682, 7.144&         40878, 38444, 38331&       22011\\ \hline

co-papers-citeseer&          2.253 , 6.214 , 11.798 &      43879 , 37352 , 37043&    26082\\ \hline

kron-g500-logn16 &           0.192 , 0.563  , 0.985&       14300 , 14176 , 14173&   14100\\ \hline

co-papers-dblp  &              3.601 , 9.422 , 9.422&       62670 , 52802 , 52196&   43978\\ \hline

\end{tabular}
\end{center}
\label{snap}
\end{footnotesize}
\end{table*}

\begin{table*}[h]
\begin{footnotesize}
\caption{Experimental results for Network repository benchmark for $m=0,2,5$.}
\begin{center}
\begin{tabular}{| llll |}
 \hline
 
\textbf{instance}&                  \textbf{time(s) $m=0,2,5$ }&              \textbf{Sol Alg \ref{alg1}} &                                      \textbf{ FastDS \cite{cai20}}\\ \hline
soc-youtube(V496k E2M)&                            4.103, 12.808, 24.979&         107611, 102464, 102366 &               89732\\ \hline
soc-flickr(V514K E3M)&                               7.190, 22.538, 48.541  &       111049, 106455, 106326 &               98062\\ \hline
ca-coauthors-dblp(V540K E15M)&                   3.679, 9.324, 17.806&          62656, 52826, 52237&                   35597\\ \hline
ca-dblp-2012(V317K E1M)   &                        2.649, 7.648, 14.942&          55071, 51946, 51841 &                  46138\\ \hline
ca-hollywood-2009(V1.1 E56.3)&                    11.049, 26.983, 46.099  &      75430, 62699, 61658&                   48740\\ \hline
inf-roadNet-CA(V2M E3M)&                            405.566, 1219.695, 2436.119 &  900458, 856218, 841233 &               586513\\ \hline
inf-roadNet-PA(V1M E2M)&                              125.252, 377.544, 750.686&     528121, 501972, 498389 &               326934\\ \hline
rt-retweet-crawl(V1M E2M)&                           12.787, 37.394, 72.781&        84722, 82386, 81350&                   75740\\ \hline
sc-ldoor(V952K E21M)&                                 16.651, 46.390, 89.056&        90654, 75095, 72627 &                  62411\\ \hline
sc-msdoor(V416K E9M) &                              2.770, 7.250, 13.904 &         28928, 25222, 23075&                   19678\\ \hline
sc-pwtk(V218K E6M)  &                                 0.714, 2.236, 4.396&           9088, 8552, 8482 &                    4200\\ \hline
sc-shipsec1(V140K E2M) &                           0.748, 1.871, 3.359&           17577, 13273, 11863&                   7662\\ \hline
sc-shipsec5(V179K E2M)&                            1.286, 3.421, 6.162&           25243, 22899, 20236 &                  10300\\ \hline
soc-FourSquare(V639K E3M)&                     18.054, 28.809, 44,000&        63643, 62245, 62050  &                 60979\\ \hline
soc-buzznet(V101K E3M) &                            0.216, 0.543, 0.932 &          201, 145, 138 &                        127\\ \hline
soc-delicious(V536K E1M)&                          3.708, 11.756, 22.703 &        58614, 57596, 55840&                   55722\\ \hline
soc-digg(V771K E6M)&                                  7.195, 19.044, 35.899 &        76454, 71018, 68846&                   66155\\ \hline
soc-flixster(V3M E8M)&                                   25.354, 73.153, 142.177&       92197, 91516, 91010&                   91019\\ \hline
soc-lastfm(V1M E5M) &                                8.326, 24.161, 45.156&         68060, 67696, 67138  &                 67226\\ \hline
soc-livejournal(V4M E28M)&                         459.524, 1230.058, 2317.355&   903243, 861908, 853903&                793887\\ \hline
soc-orkut(V3M E106M)    &                             73.094, 174.005, 307.326&      231177, 160436, 141907 &               110547\\ \hline
soc-orkut-dir(V3M E117M) &                           68.837, 163.151, 292.163&      230591, 213617, 201382&                93630\\ \hline
soc-pokec(V2M E22M)&                                 46.195, 117.618, 215.056&      301500, 241907, 239900&                207308\\ \hline
soc-youtube-snap(V1M E3M)&                       27.546, 80.038, 153.958 &      23587, 228008, 224933  &               213122\\ \hline
socfb-A-anon(V3M E24M)  &       	           79.337, 185.572, 333.808&      231903, 211315, 203021&               201690\\ \hline
socfb-B-anon(V3M E21M) &        	            70.376, 165.034, 301.235&      215288, 190114, 188315&                187030\\ \hline
socfb-FSU53(V28K E1M)&          	           0.086, 0.213, 0.333  &         3683, 2545, 2309&                      -\\ \hline
socfb-Indiana69(V30K E1M) &                 0.099, 0.253, 0.422 &          3434, 2372, 2249&                    -\\ \hline
socfb-MSU24(V32K E1M)&          	          0.099, 0.244, 0.404 &          4180,  3012,  2962 &                     -\\ \hline
socfb-Michigan23(V30K E1M)&   	              0.095, 0.237, 0.389&           4109, 2953, 2712&                    -\\ \hline
socfb-Penn94(V42K E1M) &        	            0.138, 0.344, 0.589 &          5588, 4137, 4044  &                    -\\ \hline
socfb-Texas80(V32K E1M) &       	           0.100, 0.251, 0.418&           4357, 3104, 20725 &                    -\\ \hline
socfb-Texas84(V36K E2M) &                        0.126, 0.321, 0.544&           4211, 3148, 3084  &                    -\\ \hline
rec-epinions(V755K E13M)&       	            1.114, 4.026, 8.270 &          10116, 10076, 9766 &                   9598\\ \hline
soc-dogster(V427K E9M) &      		              1.192, 2.894, 5.287 &          31000, 28300, 27261 &                  26253\\ \hline
sc-rel9(V6M E24M) &             		           257.363, 731.034, 1375.043&    230045, 201699, 197014&                127548\\ \hline
rec-libimset-dir(V221K E17M)& 	                    1.439, 4.262, 8.346 &          15229, 13910, 13225&                   12955\\ \hline

web-EPA&                     0.01, 0.012, 0.016&       361, 306, 303&   -       \\ \hline      

web-edu &                   0.015, 0.045, 0.092 &     254, 254, 253&-\\ \hline

web-polblogs&                0.003, 0.006, 0.012&      128, 117, 117&-\\ \hline

web-spam &                   0.013, 0.018, 0.026&      998, 918, 917&-\\ \hline

web-frwikinews-user-edits&                0.032, 0.074, 0.092&      696, 686, 686&-\\ \hline

web-indochina-2004 &            0.016, 0.043, 0.088 &    1519, 1516, 1516&-\\ \hline

web-webbase-2001 &              0.016, 0.043, 0.088&     1519, 1516, 1516&   -   \\ \hline          

web-sk-2005 &              0.983, 2.662, 5.136&	  33767, 32345, 32305 &  -   \\ \hline       

web-uk-2005&	              0.149, 0.496, 0.864&	   1719, 1719, 1719&                1421\\ \hline

web-arabic-2005&            0.639, 2.339, 3.781&	   21261, 20307, 20279&-\\ \hline

web-Stanford&	             0.392, 1.456, 3.046 &    21335, 19961, 19942&-\\ \hline

web-NotreDame &            4.039, 8.443, 14.713&	  1197991, 1196608, 1196579&-\\ \hline

web-BerkStan-dir&          10.358, 24.594, 42.32 &   6959022, 6955884, 6955817&-\\ \hline

web-it-2004&                 4.193, 12.689, 25.002&   34770, 34497, 34442  &            32997\\ \hline

web-italycnr-2000&         1.832, 4.96, 9.72 &     26137, 24554, 24530&-\\ \hline

web-wikipedia2009&         103.305, 297.81, 590.188&    421852, 400294, 399634  &         346581\\ \hline

web-google-dir&             12.913, 32.988, 59.198&  4336976, 4325727, 4325505&-\\ \hline

web-baidu-baike &         45.817, 127.943, 249.043&   326856, 310844, 310441&       277847\\ \hline

\end{tabular}
\end{center}
\label{nr}
\end{footnotesize}
\end{table*}
These experiments show that the fast running time and easy implementation of our proposed algorithm enable us to compute reasonable solutions for really large networks.
 This algorithm can be used in very large networks that previous algorithms were not able to compute the solution because of time or hardware limitations.

The quality of solutions in our algorithm is worse than \cite{cai20} in some benchmarks. however our results are at most 2 times of the results of FastDS which is reasonable since as we said before our algorithm finds a total dominating set which potentially can be 2 times larger than the dominating set that the other algorithms find and also the algorithm is distributed and the previous ones do not work in distributed setting. Obviously not all the sequential algorithms can be made distributed with the same solution quality. This happens in distributed algorithms for MDS problem as well. We wanted to show that a simple and fast distributed algorithm can have reasonable solutions.
However because of the fast running time of our algorithm, it can be use in sequential model in the case where we have time limit, since even for small values of $m$, the algorithm gives reasonable results.

For the further experiments we solve the $k$-distance dominating set problem on some of the large social networks instances. The instances are from the network repository benchmark.
We compare our results with the results of \cite{min}. Their experiments are conducted on a computer with Intel Core i7-8750h 2.2 GHz running Ubuntu OS. The programming language is Python using igraph package to perform graph computations. In \cite{min} their algorithm known as HEU4 was compared with an algorithm called HEU3 from \cite{cam} and another greedy algorithm called HEU1 which was used by their partner.
In Table \ref{k} the results are shown. As it can be seen our algorithm has reasonable results in distributed model compared to their results in sequential model.
Note that again we cannot compare the running times since the configuration of systems are different.

If we construct a graph $G'$ from $G$ such that there is an edge between two vertices $u$ and $v$ if their distance is at most $k$, then solving the minimum $k$-distance dominating set for $G$ is equal to solving MDS for $G'$. This means that for larger values of $k$ the graph $G'$ will be denser than $G$. So we expect that our algorithm outputs near optimal solutions for larger values of $k$ in dense graphs. In the $k$-distance dominating set problem theoretically the algorithm needs $O(km)$ rounds because when a node wants to mark its neighbor with maximum $w_i$, it needs $k$ rounds to receive $w_i$'s of neighbors within distance $k$ from itself.
\begin{table*}[h]
\begin{footnotesize}
\caption{Experimental results for Network repository benchmark for $m=0,2,5$ and $k=2$.}
\begin{center}
\begin{tabular}{| lllllll |}
 \hline

  \textbf{ instance}&          \textbf{Alg \ref{alg1} sol m=0,2,5}&        \textbf{Alg \ref{alg1} time(s) m=0,2,5}&                                      \textbf{ HEU1 sol}&           \textbf{HEU1 time(s)} &                        \textbf{   HEU4 sol}&          \textbf{HEU4 time}\\ \hline
soc-delicious&        32029, 18666, 17575&         6.763, 18.937, 36.800&                     34516 &                3.57  &                 8155&                                 2064.00\\ \hline
soc-flixster  &        46911, 20981, 17431&     152.383, 413.453, 776.118&                 48789 &                85.93&                              9860&                  23694.58\\ \hline
soc-livejournal  &     395992, 291846, 277713 &     269.600, 656.184, 1182.615&             447552 &               1582.87&                            189121&                16728.22\\ \hline

\end{tabular}
\end{center}
\label{k}
\end{footnotesize}
\end{table*}

Note that we have implemented the sequential version of Algorithm \ref{alg1} without any change on the distributed version.
In the distributed model all nodes are executing at the same time, here at the same time in each round they choose their neighbor with maximum $w_i$, in our sequential implementation the same nodes are chosen but one by one as well.
So the results of distributed implementation and sequential implementation are the same.

We can modify this centralized version to improve the quality of the solution and also reduce the running time.
For example in the {\it{for}} loop that for each vertex $v_i$ we choose a neighbor with maximum weight, we can have the following modification. For each vertex $v_i$ if it is not dominated choose a neighbor $u$ with maximum weight and dominate the vertices that are adjacent to $u$.
This modification can improve the quality of solution and the running time since for the vertices that are dominated in the previous steps we do not need to find the neighbor with maximum weight. Also for the running time we have not tried to use fast algorithms for finding the neighbors within distance $k$.
But we have not done the above modifications because we want to show the efficiency and easy implementation of Algorithm \ref{alg1} in the distributed model.

\section{Concluding remarks}
In this paper we proposed a distributed algorithm for the MDS problem and also $k$-distance dominating set problem.
Theoretically, for MDS problem we analyzed the approximation factor of Algorithm \ref{alg1} for triangle free planar graphs. 
As it can be seen in the proofs, the restriction for not having triangles or $4$-cycles (in \cite{ali}) is used to get a graph without multiple edges. Hence for a planar graph $G$ or even a general graph $G$ with small crossing number, we can still use this algorithm to get a possibly worse approximation factor. The crossing number  of a graph $G$ is the lowest number of edge crossings of a plane drawing of the graph $G$. For instance, a graph is planar if and only if its crossing number is zero.
In practice as an example of such networks  consider the network of streets in an urban area where the streets are edges and the cross points of streets are nodes. This graph is  planar graph and it  has a few number of triangles.

Also, as the experiments show by increasing $m$ in our algorithm the quality of the solution increases. It is an interesting problem to theoretically see how the approximation factors given in Theorem \ref{main} change if one increases $m$. The experiments suggest that for small values of $m$ this improvement is significant and as $m$ increases to bigger numbers, the improvements are marginal.

Also in a future work, one can modify the algorithm to solve the set-cover problem. In the set-cover problem we are given a set $A=\{a_1,a_2,\dots, a_n\}$ of $n$ elements and $m$ subsets, $A_1,A_2,\dots,A_m$ of $A$. The goal is to choose the minimum number of subsets that they cover all the elements of $A$. Again simply each element $a_i$ chooses a  subset $A_j$ with maximum size such that $a_i\in A_j$. Then $x_i$'s are the number of times that $A_i$'s are chosen by the elements. Then for $m$ rounds we repeat the algorithm.

Note that beside the fact that this algorithm is distributed we can easily implement it in a dynamic model where the network is changing dynamically. This is the case that happens in social networks constantly. When an edge is added to the network, we can easily update the answer based on the added vertex. Specially in the case of $m=0$, it is enough to change the answer of those vertices that are connected to the added vertex and their neighbors. So the update can be done very fast.

Here in the first step of algorithm, we consider the neighbor vertices with maximum degrees. However based on the input graph one can modify the algorithm and choose the initial vertices based on other properties.

\bibliographystyle{ACM-Reference-Format}
\bibliography{sample-base}

\end{document}